\def\year{2022}\relax
\documentclass[letterpaper]{article} 
\usepackage{aaai22}  
\usepackage{times}  
\usepackage{helvet}  
\usepackage{courier}  
\usepackage[hyphens]{url}  
\usepackage{graphicx} 
\usepackage{natbib}  
\usepackage{caption} 
\DeclareCaptionStyle{ruled}{labelfont=normalfont,labelsep=colon,strut=off} 
\usepackage{booktabs}       
\usepackage{amsfonts}       
\usepackage{nicefrac}       
\usepackage{microtype}      
\usepackage{subcaption}
\usepackage{amsthm}
\usepackage{algorithm, algpseudocode}
\usepackage{amstext} 
\usepackage{array}   
\newcolumntype{L}{>{$}l<{$}} 
\newcolumntype{C}{>{$}c<{$}} 
\usepackage{siunitx} 
\usepackage{multirow}

\usepackage{wrapfig,lipsum,booktabs}

\newtheorem{assumption}{Assumption}
\newtheorem{problem}{Problem}

\newcommand{\R}{\mathbb{R}}
\newcommand{\bmat}[1]{\begin{bmatrix}#1\end{bmatrix}}
\newcommand{\smat}[1]{\left[\begin{smallmatrix} #1 \end{smallmatrix} \right]}

\usepackage{newfloat}
\usepackage{listings}
\floatstyle{ruled}
\newfloat{listing}{tb}{lst}{}
\floatname{listing}{Listing}
\title{Recurrent Neural Network Controllers Synthesis with Stability Guarantees for Partially Observed Systems}
\author{
    Fangda Gu\equalcontrib\textsuperscript{\rm 1}, He Yin\equalcontrib\textsuperscript{\rm 1}, Laurent {El Ghaoui}\textsuperscript{\rm 1}, Murat Arcak\textsuperscript{\rm 1}, Peter Seiler\textsuperscript{\rm 2}, Ming Jin\textsuperscript{\rm 3}
}
\affiliations{
    \textsuperscript{\rm 1} University of California, Berkeley, 2594 Hearst Ave, Berkeley, California 94720\\
    \textsuperscript{\rm 2} University of Michigan, 500 S State St, Ann Arbor, Michigan  48109\\
    \textsuperscript{\rm 3} Virginia Tech, 1185 Perry Street 453 Whittemore (0111), Blacksburg, Virginia 24061\\

}

\usepackage{lipsum}
\usepackage{amsfonts}
\usepackage{amsmath}
\usepackage{amssymb}
\usepackage{amsthm}
\usepackage{amsopn}
\usepackage{tikz}
\usetikzlibrary{arrows,calc,positioning,shadows,shapes,matrix,fit}
\usepackage{verbatim}
\usepackage{todonotes}

\newcommand{\BEAS}{\begin{eqnarray*}}
\newcommand{\EEAS}{\end{eqnarray*}}
\newcommand{\BEA}{\begin{eqnarray}}
\newcommand{\EEA}{\end{eqnarray}}
\newcommand{\BEQ}{\begin{equation}}
\newcommand{\EEQ}{\end{equation}}
\newcommand{\BIT}{\begin{itemize}}
\newcommand{\EIT}{\end{itemize}}
\newcommand{\BNUM}{\begin{enumerate}}
\newcommand{\ENUM}{\end{enumerate}}

\newcommand{\BA}{\begin{array}}
\newcommand{\EA}{\end{array}}

\newcommand{\eg}{{\it e.g.}}
\newcommand{\ie}{{\it i.e.}}

\newcommand{\diag}{\mathop{\bf diag}}

\newtheorem{definition}{Definition}[section]
\newtheorem{theorem}{Theorem}[section]
\newtheorem{lemma}{Lemma}[section]
\newtheorem{remark}{Remark}[section]

\usepackage{enumitem}

\newcommand\blfootnote[1]{%
  \begingroup
  \renewcommand\thefootnote{}\footnote{#1}%
  \addtocounter{footnote}{-1}%
  \endgroup
}

\begin{document}

\maketitle


\begin{abstract}
Neural network controllers have become popular in control tasks thanks to their flexibility and expressivity. Stability is a crucial property for safety-critical dynamical systems, while stabilization of partially observed systems, in many cases, requires controllers to retain and process long-term memories of the past. We consider the important class of recurrent neural networks (RNN) as dynamic controllers for nonlinear uncertain partially-observed systems, and derive convex stability conditions based on integral quadratic constraints, S-lemma and sequential convexification. 
To ensure stability during the learning and control process, we propose a projected policy gradient method that iteratively enforces the stability conditions in the reparametrized space taking advantage of mild additional information on system dynamics. Numerical experiments show that our method learns stabilizing controllers while using fewer samples and achieving higher final performance compared with policy gradient.


\end{abstract}

\section{Introduction}


Neural network decision making and control has seen a huge advancement recently accompanied by the success of reinforcement learning (RL) \citep{sutton2018reinforcement}. In particular, deep reinforcement learning (DRL) has achieved super-human performance with neural network policies (also referred to as controllers in control tasks) in various domains \citep{mnih2015human, lillicrap2016continuous,silver2016mastering}. 

Policy gradient \citep{sutton1999policy} is one of the most important approaches to DRL that synthesizes policies for continuous decision making problems. For control tasks, policy gradient method and its variants have successfully synthesized neural network controllers to accomplish complex control goals \citep{levine2018learning} without solving potentially non-linear planing problems at test time \citep{levine2016end}. 
However, most of these methods 
focus on maximizing the reward function which only indirectly enforce desirable properties. Specifically, global stability of the closed-loop system \citep{sastry2013nonlinear} guarantees 
convergence to the desired state of origin from any initial state and therefore is a very important property for safety critical systems     (\eg\; aircraft control \citep{chakraborty2011susceptibility}) where not a single diverging trajectory is acceptable.
However, the set of parameters corresponding to stabilizing controllers is in general nonconvex even in the simple setting of linear systems with linear controllers \cite{fazel2018global}, which poses significant computational challenges for neural network controllers under the general setting of nonlinear systems.

Thanks to recent robustness studies of deep learning, 
we have seen attempts on giving stability certificates 
and/or ensuring stability at test time 
for fully-observed systems controlled by neural networks. Yet stability problems for neural network controlled partially observed systems remain open. Unlike fully-observed control systems 
where the plant states are fully revealed to the controller, most real-world control systems are only partially observed 
due to modeling inaccuracy, sensing limitations, and physical constraints \citep{braziunas2003pomdp}.  
Here, sensible estimates of the full system state usually depend on historical observations \citep{callier2012linear}. Some partially observed systems are modeled using partially observed Markov decision process (POMDP) \citep{monahan1982state} where an optimal solution is NP hard in general \citep{mundhenk2000complexity}.


\textbf{Paper contributions.}
%
In the paper, we propose a method to synthesize recurrent neural network (RNN) controllers with exponential stability guarantees for partially observed systems. 
We derive a convex inner approximation to the non-convex set of stable RNN parameters based on integral quadratic constraints \citep{megretski1997system}, loop transformation \citep[Chap. 4]{sastry2013nonlinear} and a sequential semidefinite convexification technique, which guarantees exponential stability for both linear time invariant (LTI) systems and general nonlinear uncertain systems.
A novel framework of projected policy gradient is proposed to maximize some unknown/complex reward function and ensure stability in the online setting  
where a guaranteed-stable RNN controller is synthesized and iteratively updated while interacting with and controlling the underlying system, which differentiates our works from most post-hoc validation methods.
Finally, we carry out comprehensive comparisons with policy gradient, and demonstrate that our method effectively ensures the closed-loop stability and achieves higher reward on a variety of control tasks, including vehicle lateral control and power system frequency regulation. 

\begin{figure*}[h!]
  \centering
  \begin{minipage}{0.49\textwidth}
  \centering
  \includegraphics[width=0.85\textwidth]{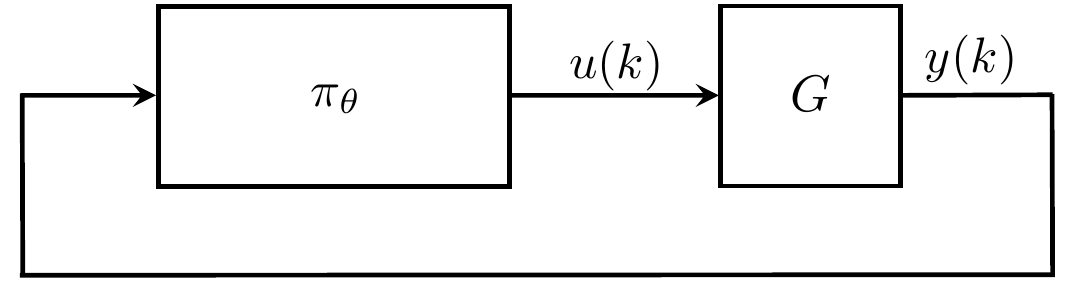}
  \caption{Feedback system of plant $G$ and RNN controller $\pi_\theta$}
  \label{fig:feedback}    
  \end{minipage}
  \hfill
  \begin{minipage}{0.45\textwidth}
  \centering
  \includegraphics[width=0.8\textwidth]{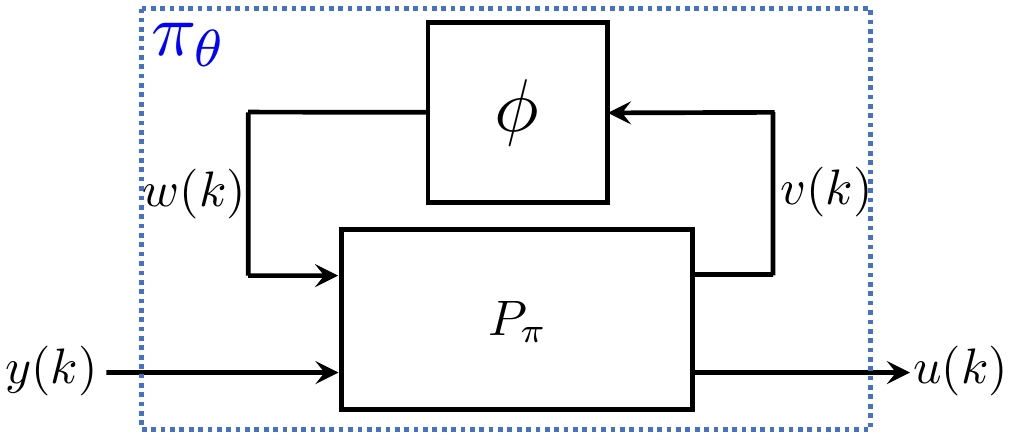}
  \caption{RNN  as an interconnection of $P_\pi$ and $\phi$}
  \label{fig:RNN}    
\end{minipage}
\end{figure*}

\textbf{Paper outline.} 
In Section~\ref{sec:related_work}, we outline related works on addressing partial observability, and enforcing stability in reinforcement learning. 
Section~\ref{sec:LTI} discusses our proposed method for synthesizing RNN controllers for LTI plants with stability guarantees, and Section~\ref{sec:IQC} extends it to systems with uncertainties and nonlinearities. Section~\ref{sec:numerical}  compares the proposed projected policy gradient method with policy gradient through numerical experiments.

\textbf{Notation.}
$\mathbb{S}^n, \mathbb{S}^n_{+}, \mathbb{S}^n_{++}$ denote the sets of $n$-by-$n$ symmetric, positive semidefinite and positive definite matrices, respectively. $\mathbb{D}^n_{+}$, $\mathbb{D}^n_{++}$ denote the set of $n$-by-$n$ diagonal positive semidefinite, and diagonal positive definite matrices.  The notation $\|\cdot\|: \R^n \rightarrow \R$ denotes the standard 2-norm. We define $\ell_{2e}^n$ to be the set of all one-sided sequences $x: \mathbb{N} \rightarrow \R^n$. The subset $\ell_{2}^{n} \subset \ell_{2e}^n$ consists of all square-summable sequences. When applied to vectors, the orders $>, \leq$ are applied elementwise.

\section{Related Work} \label{sec:related_work}
\textbf{Partially Observed Decision Making and Output Feedback Control.}
In many problems \citep{talpaert2019exploring, barto1995learning}, only specific outputs but not the full system states are available for the decision maker. 
Therefore, memory in the controller is required to recover the full system states \citep{scherer1997multiobjective}. Control of these partially observed systems is often referred to as output feedback control \citep{callier2012linear}, and has been studied extensively from both control and optimization perspectives \citep{doyle1978guaranteed, zheng2021analysis}. Under the setting with convexifiable objectives (e.g., $H_{\infty}$ or $H_2$ performances), the optimal linear dynamic (\ie~ with memory) controller can be obtained by using a change of variables or solving algebraic Riccati equations \citep{gahinet1994linear, zhou1996robust}. However, for more sophisticated settings with unknown and/or flexibly defined cost functions, the problems  become intractable for the aforementioned traditional methods, and RL techniques
are proposed to reduced the computation cost and improve overall performance at test time, including the ones \citep{levine2013guided, levine2016end} with static neural network controllers, and the ones \citep{zhang2016learning, heess2015memory, wierstra2007solving} with dynamic controllers, represented by RNNs/long short-term memory neural networks. 


\textbf{Stability Guarantees For Neural Network Controlled Systems.}
As neural networks become popular in control tasks, safety and robustness of neural networks and neural network controlled systems has been actively discussed \citep{morimoto2005robust,  luo2014off, friedrich2017robust, berkenkamp2017safe, chow2018lyapunov, matni2019self, han2019h, recht2019tour, choi2020reinforcement, zhang2020policy, fazlyab2020safety}. Closely related to this work are recent papers on robustness analysis of memory-less neural networks controlled systems based on robust control ideas. 
\citet{yin2021stability, yin2021imitation, pauli2021linear, jin2020stability}  
conduct stability analysis of neural network controlled linear and nonlinear systems and propose verification methods by characterizing activation functions using quadratic constraints. \citet{donti2020enforcing} adds additional projection layer on the controller to ensure stability for fully observed systems. 
\citep{revay2020convex} studies the stability of RNN itself when fitted to data but does not consider any plant to control by such RNN. The most related works are those that study dynamic neural network controllers. \citet{anderson2007robust, knight2011stable} adapt RNN controllers through RL techniques to obtain stability guarantees. However, in these works, the reward function is assumed to be known, and conservative updates of controller parameters projected to a box neighborhood of the previous iterate are applied due to the non-convexity in their conditions. In contrast, our work enables much larger and more efficient updates thanks to jointly convex conditions derived through a novel sequential convexification and loop transformation
approach unseen in these works. 




\begin{figure*}[h!]
  \centering
  \begin{minipage}{0.49\textwidth}
  \centering 
		\includegraphics[width=0.7\textwidth]{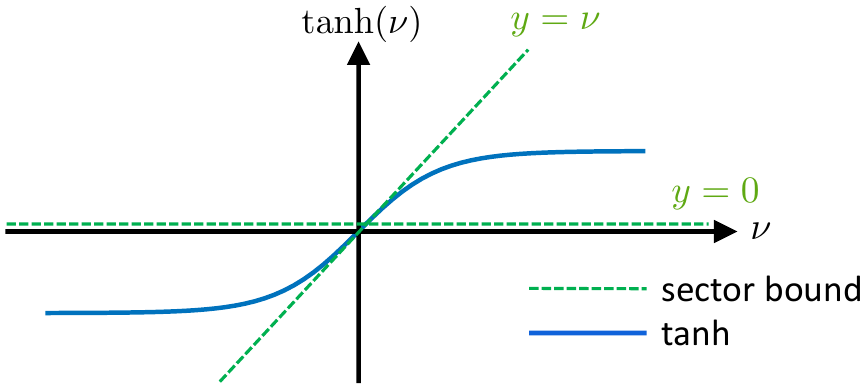}
		\includegraphics[width=0.7\textwidth]{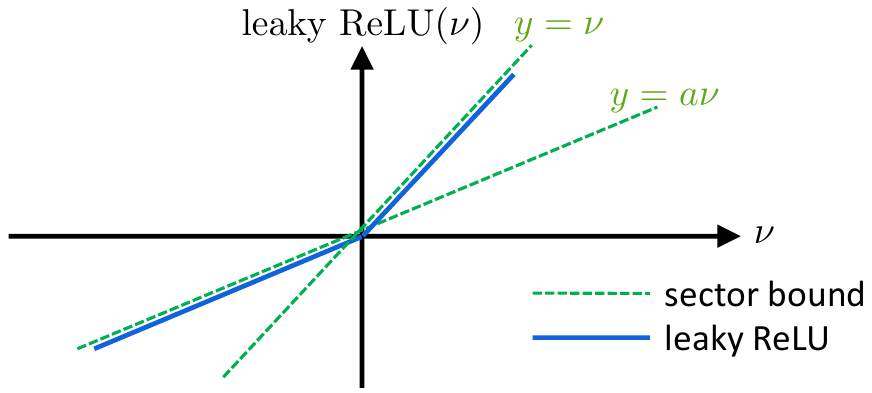}
		\caption{$\tanh \in$ sector $[0, 1]$,  \\
		Leaky ReLU $\in$ sector $[a, 1]$}
		\label{fig:sector}
  \end{minipage}
  \hfill
  \begin{minipage}{0.49\textwidth}
  \centering 
  \includegraphics[width=0.88\textwidth]{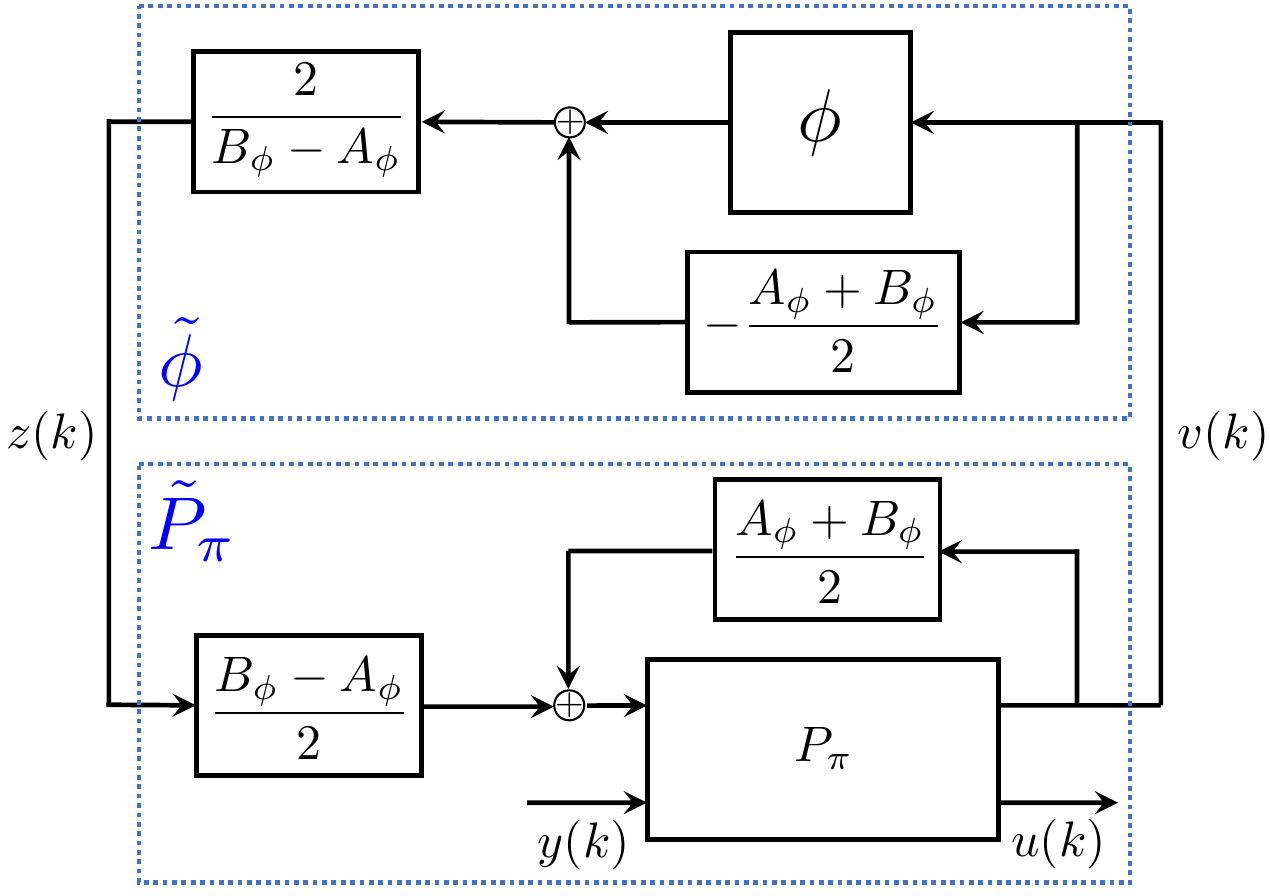}
  \caption{Loop transformation.  If $\phi \in$ sector $[\alpha_\phi, \beta_\phi]$, then $\tilde{\phi} \in $ sector $[-1_{n_\phi \times 1}, 1_{n_\phi\times 1}]$.}
  \label{fig:loop_tran}    
  \end{minipage}
\end{figure*} 

\section{Partially Observed Linear Systems} \label{sec:LTI}
\subsection{Problem Formulation}
Consider the feedback system (shown in Fig.~\ref{fig:feedback}) consisting of a plant $G$ and an RNN controller $\pi_\theta$ which is expected to stabilize the system (\ie~steer the states of $G$ to the origin). To streamline the presentation, we consider a partially observed, linear, time-invariant (LTI) system~$G$ defined by the following discrete-time model:
\begin{subequations}\label{eq:nomi_G}
\begin{align}
x(k+1) &= A_G \ x(k) + B_G \ u(k) \\
y(k) &= C_G \ x(k)
\end{align}
\end{subequations}
where $x(k) \in \R^{n_G}$ is the state, $u(k) \in \R^{n_u}$ is the control input, and $y(k) \in \R^{n_y}$ is the output. $A_G \in \R^{n_G \times n_G}$, $B_G \in \R^{n_G \times n_u}$, and $C_G \in \R^{n_y \times n_G}$.  Since the plant $G$ is partially observed, the observation matrix $C_G$ may have a sparsity pattern or be column-rank deficient.  
\begin{assumption}
We assume that $(A_G, B_G)$ is stabilizable, and $(A_G, C_G)$ is detectable\footnote{The definitions of stabilizability and detectability can be found in \citep{callier2012linear}.}.
\end{assumption}

\begin{assumption}\label{ass:known_dyn}
We assume $A_G, B_G,$ and $C_G$ are known.
\end{assumption}
Assumption~\ref{ass:known_dyn} is partially lifted in Section~\ref{sec:IQC} where we only assume partial information on the system dynamics.

\begin{problem}
Our goal is to find a controller $\pi$ that maps the observation $y$ to an action $u$ to both maximize some unknown reward $R = \sum_{k=0}^T r_k(x(k), u(k))$ over finite horizon $T$ and stabilize the plant $G$.
\end{problem}

The single step reward $r_k(x(k), u(k))$ is assumed to be unknown and potentially highly complex to capture the vast possibility of desired controller. e.g. In many cases, to ensure extra safety, the reward is set to $r_k(x(k), u(k)) = 0, \forall k \ge l$ if there is a state violation at step $l$. This cannot be captured by any simple negative quadratic functions.


\subsection{Controllers Parameterization}
Output feedback control with known and convexifiable reward has been studied extensively \citep{scherer1997multiobjective}, and linear dynamic controllers suffice for this case. However, in our problem setting, since the reward is unknown and nonconvex, and systems dynamics will become uncertain and nonlinear in Section~\ref{sec:IQC}, we consider a dynamic controller in the form of an RNN, which makes a class of high-capacity flexible controllers.

We model the RNN controller $\pi_\theta$ as an interconnection of an LTI system $P_\pi$, and combined activation functions $\phi: \R^{n_\phi} \rightarrow \R^{n_\phi}$ as shown in Fig.~\ref{fig:RNN}. This parameterization is expressive, and contains many widely used model structures \citep{revay2020convex}. The RNN $\pi_\theta$ is defined as follows 
\begin{align} 
&P_\pi \left\{\begin{array}{ll}
\xi(k+1) &= A_K \ \xi(k) \ + B_{K1} \ w(k) + B_{K2} \ y(k) \\
u(k) &= C_{K1} \ \xi(k) + D_{K1} \ w(k) + D_{K2}  \ y(k) \\
v(k) &= C_{K2} \ \xi(k) + D_{K3} \ y(k) \end{array}\right. \notag \\
&\hspace{0.87cm} w(k) \hspace{0.85cm} = \phi(v(k)) \label{eq:PK_def}
\end{align}
where $\xi \in \R^{n_\xi}$ is the hidden state, $v, w \in \R^{n_\phi}$ are the input and output of $\phi$, and matrices $A_K, \dots, D_{K3}$ are parameters to be learned. Define $\theta = \smat{A_K & B_{K1} & B_{K2} \\ C_{K1} & D_{K1} & D_{K2} \\ C_{K2} & 0 & D_{K3}}$ as the collection of the learnable parameters of $\pi_\theta$. We assume the initial condition of $\xi$ to be zero $\xi(0) = 0_{n_\xi \times 1}$. 
The combined nonlinearity $\phi$ is applied element-wise, i.e., $\phi := [\varphi_1(v_1), ..., \varphi_{n_\phi}(v_{n_\phi})]^\top$, where $\varphi_i$ is the $i$-th scalar activation function. We assume that the activation has a fixed point at origin, \ie~$\phi(0) = 0$.





\subsection{Quadratic Constraints for Activation Functions}
The stability condition relies on quadratic constraints (QCs)
to bound the activation function. A typical QC is the
sector bound as defined next.
\begin{definition}
  \label{def:sector}
  Let $\alpha \le \beta$ be given.  The function
  $\varphi: \R \rightarrow \R$ lies in the sector
  $[\alpha,\beta]$ if:
  \begin{align}
    ( \varphi(\nu) - \alpha \nu ) \cdot
       (\beta \nu - \varphi(\nu)) \ge 0
    \,\,\, \forall \nu \in \R.
  \end{align}
\end{definition}
The interpretation of the sector $[\alpha,\beta]$ is  that
$\varphi$ lies between lines passing through the origin with slope
$\alpha$ and $\beta$. Many activations are sector bounded, e.g., leaky ReLU is sector bounded in $[a, 1]$ with its parameter $a \in (0, 1)$; ReLU and $\tanh$ are sector bounded in $[0,1]$ (denoted as $\tanh \in $ sector $[0, 1]$). Fig.~\ref{fig:sector} illustrates different activations (blue solid) and their sector bounds (green dashed).

Sector constraints can also be defined for combined activations $\phi$. Assume the $i$-th scalar activation $\varphi_i$ in $\phi$ is sector bounded by $[\alpha_i, \beta_i]$, $i=1,...,n_\phi$, then these sectors can be stacked into vectors $\alpha_\phi, \beta_\phi \in \R^{n_\phi}$, where $\alpha_\phi=[\alpha_1,...,\alpha_{n_\phi}]$ and $\beta_\phi=[\beta_1, ..., \beta_{n_\phi}]$, to provide QCs satisfied by  $\phi$.
\begin{lemma}
Let $\alpha_\phi, \beta_\phi \in \R^{n_\phi}$ be given with $\alpha_\phi \leq \beta_\phi$. Suppose that $\phi$ satisfies the sector bound $[\alpha_\phi, \beta_\phi]$ element-wise. For any $\Lambda \in \mathbb{D}_{+}^{n_\phi}$, and for all $v \in \R^{n_\phi}$ and $w = \phi(v)$, it holds that
\begin{align}
        \bmat{v \\ w}^\top \bmat{-2A_\phi B_\phi \Lambda & (A_\phi+B_\phi)\Lambda \\ (A_\phi+B_\phi)\Lambda & -2\Lambda} \bmat{v \\ w} \ge 0, \label{eq:origin_QC}
\end{align}
where $A_\phi = diag(\alpha_\phi)$, and $B_\phi = diag(\beta_\phi)$. 
\end{lemma}
A proof is available in \citep{fazlyab2020safety}. 

\subsection{Loop Transformation}
To derive convex stability conditions for their efficient enforcement in
the learning process, we first perform a loop transformation on the RNN as shown in Fig.~\ref{fig:loop_tran}.
Through loop transformation, we obtain a new representation of the controller $\pi_{\tilde{\theta}}$, which is equivalent to the one shown in Fig.~\ref{fig:RNN}: \begin{subequations}\label{eq:RNN_transformed}
\begin{align}
    \bmat{v \\ u} &= \tilde{P}_\pi \bmat{z \\ y} \\
    z(k) &= \tilde{\phi} (v(k)).
\end{align}
\end{subequations}
The newly obtained nonlinearity $\tilde{\phi}$, defined in Fig.~\ref{fig:loop_tran}, is sector bounded by $[-1_{n_\phi \times 1}, 1_{n_\phi \times 1}]$, and thus it satisfies a simplified QC: for any $\Lambda \in \mathbb{D}_{+}^{n_\phi}$, it holds that
\begin{align}
    \bmat{v \\ z}^\top \bmat{\Lambda & 0 \\ 0 & -\Lambda} \bmat{v \\ z} \ge 0, \ \forall v \in \R^{n_\phi} \ \text{and} \ z = \tilde{\phi}(v). \label{eq:shifted_QC}
\end{align}

The transformed system $\tilde{P}_\pi$, defined in Fig.~\ref{fig:loop_tran}, is of the form:
\begin{align*}
\begin{array}{ll}
    \xi(k+1) &= \tilde{A}_K ~~ \xi(k) + B_{K1}\frac{B_\phi - A_\phi}{2} \ z(k) + \tilde{B}_{K2} \ y(k) \\
u(k) &= \tilde{C}_{K1} \ \xi(k) + D_{K1}\frac{B_\phi - A_\phi}{2} \ z(k) + \tilde{D}_{K2} \ y(k) \\
v(k) &= C_{K2} \ \xi(k) + D_{K3} \ y(k) \end{array}
\end{align*}
where
\small{
\begin{align}
&\begin{array}{l}
\hspace{0.12cm}\tilde{A}_K = A_K + B_{K1} S_\phi C_{K2}, \hspace{0.3cm} \tilde{B}_{K2} = B_{K2} + B_{K1} S_\phi D_{K3}, \\
\tilde{C}_{K1} = C_{K1} + D_{K1} S_\phi C_{K2}, \hspace{0.15cm} \tilde{D}_{K2} = D_{K2} + D_{K1} S_\phi D_{K3},  
\end{array}\notag \\
&\hspace{3cm} S_\phi := \textstyle \frac{A_\phi + B_\phi}{2}. \label{eq:1to1correspond}
\end{align}
}
The derivation of $\tilde{P}_\pi$ can be found in Appendix~\ref{app:tranplant}. We define the learnable parameters of $\pi_{\tilde{\theta}}$ as $\tilde{\theta} = \smat{\tilde{A}_K & B_{K1} & \tilde{B}_{K2} \\ \tilde{C}_{K1} & D_{K1}& \tilde{D}_{K2} \\  C_{K2} & 0 & D_{K3}}$. 
Since there is an one-to-one correspondence \eqref{eq:1to1correspond} between the transformed parameters $\tilde{\theta}$ and the original parameters $\theta$, we will learn in the reparameterized space and uniquely recover the original parameters accordingly.

\subsection{Convex Lyapunov Condition}
The feedback system of plant $G$ and RNN controller in $\pi_{\tilde{\theta}}$  \eqref{eq:RNN_transformed} is defined by the following equations
\begin{subequations}\label{eq:feedback_nominal}
\begin{align}
    \zeta(k+1) &= \mathcal{A} \ \zeta(k) + \mathcal{B} \ z(k)\\
v(k) &= \ \mathcal{C} \ \zeta(k) + \mathcal{D} \ z(k) \\
z(k) & = \tilde{\phi}(v(k))
\end{align}
\end{subequations}
where $\zeta = [x^\top, \ \xi^\top] ^\top$ gathers the states of $G$ and $\tilde{P}_\pi$, and
\begin{align*}
\mathcal{A} &= \smat{ A_G + B_G \tilde{D}_{K2}C_G & B_G \tilde{C}_{K1} \\ \tilde{B}_{K2}C_G & \tilde{A}_{K}}, \mathcal{B} = \smat{B_G D_{K1}\frac{B_\phi - A_\phi}{2} \\ B_{K1}\frac{B_\phi - A_\phi}{2}},\\
\mathcal{C} &= \smat{D_{K3}C_G & C_{K2}}, \hspace{1.5cm} \mathcal{D} = 0_{n_\phi\times n_\phi}.
\end{align*}
Note that matrices $\mathcal{A}, \mathcal{B}, \mathcal{C}, \mathcal{D}$ are affine in $\tilde{\theta}$. The following theorem incorporates the QC for $\tilde{\phi}$ in the Lyapunov condition to derive the exponential stability condition of the feedback system using the S-Lemma \citep{yakubovich1971, boyd1994linear}

\begin{theorem} [Sequential Convexification] \label{thm:lyap_nominal}
Consider the feedback system of plant $G$ in \eqref{eq:nomi_G}, and RNN controller $\pi_{\tilde{\theta}}$ in \eqref{eq:RNN_transformed}. Given a rate $\rho$ with $0 \leq \rho \leq 1$, and matrices $\bar{P} \in \R^{n_\zeta \times n_\zeta}$ and $\bar{\Lambda} \in \R^{n_\phi \times n_\phi}$, if there exist matrices $Q_1 \in \mathbb{S}_{++}^{n_\zeta}$ and $Q_2 \in \mathbb{D}_{++}^{n_\phi}$, and parameters $\tilde{\theta}$ such that the following condition holds
\small{
\begin{align}
     \begin{bmatrix}\rho^2(2\bar{P}-\bar{P}^\top Q_1 \bar{P}) & 0 & \mathcal{A}^\top & \mathcal{C}^\top \\ 0 & 2\bar{\Lambda} - \bar{\Lambda}^\top Q_2 \bar{\Lambda} & \mathcal{B}^\top& \mathcal{D}^\top \\ \mathcal{A} & \mathcal{B} & Q_1 & 0\\ \mathcal{C} & \mathcal{D} & 0 &Q_2 \end{bmatrix} \succeq 0, \label{eq:lyap_convex}
\end{align}
}
then for any $x(0)$, we have $\|x(k)\| \leq \sqrt{\text{cond}(P)}\rho^k \|x(0)\|$ for all $k$, where cond$(P)$ is the condition number of $P$, and $P := Q_1^{-1}$ i.e., the feedback system is exponentially stable with rate $\rho$. 
\end{theorem}
The above convex relaxation of the non-convex condition \eqref{eq:lyap_cond} leverages a ``linearizing'' semi-definite inequality based on a previous guess of $Q_1^{-1}$ and $Q_2^{-1}$ (as $\bar{P}$ and $\bar{\Lambda}$). A complete proof is provided in Appendix~\ref{app:pf_nominal_stab}. The linear matrix inequality (LMI) condition \eqref{eq:lyap_convex} is jointly convex in the decision variables $\tilde{\theta}$, $Q_1$, $Q_2$, where $Q_1$ and $Q_2$ are the inverse matrices of the Lyapunov certificate and the multiplier in \eqref{eq:lyap_cond}, and this allows for its efficient enforcement in the reinforcement learning process.
Denote the LMI~\eqref{eq:lyap_convex}, $Q_1 \in \mathbb{S}_{++}^{n_\zeta}$, and $Q_2 \in \mathbb{D}_{++}^{n_\phi}$ altogether as LMI$(Q_1,Q_2,\tilde{\theta}, \bar{P}, \bar{\Lambda})$, which will later be incorporated in the policy gradient process to provide exponential stability guarantees. 

Based on the stability condition \eqref{eq:lyap_convex}, define the convex stability set $\mathcal{C}(\bar{P}, \bar{\Lambda})$ as 
\begin{align}
\mathcal{C}(\bar{P}, \bar{\Lambda}) := \left\{\tilde{\theta}: \ \exists \ Q_1, Q_2, \ \text{s.t. LMI}(Q_1, Q_2, \tilde{\theta}, \bar{P}, \bar{\Lambda})\right\}.
\end{align}
Given matrices $\bar{P}$ and $\bar{\Lambda}$, any parameter $\tilde{\theta}$ drawn from $\mathcal{C}(\bar{P}, \bar{\Lambda})$ ensures the exponential stability of the feedback system \eqref{eq:feedback_nominal}. The set $\mathcal{C}(\bar{P}, \bar{\Lambda})$ is a convex inner-approximation to the set of parameters that renders the feedback system stable, and the choice of $\bar{P}$ and $\bar{\Lambda}$ affects the conservatism in the approximation. One way of choosing $(\bar{P}, \bar{\Lambda})$ is provided in Algorithm~\ref{alg:alg1}.
\begin{remark} \label{rem:sector}
Although only sector bounds \eqref{eq:shifted_QC} are used to describe the activation functions, we can further reduce the conservatism by using  off-by-one integral quadratic constraints \citep{lessard2016analysis} to also capture the slope information of activation functions as done in \citep{yin2021stability}.
\end{remark}

\begin{remark} \label{rem:lti}
Note that although we only consider LTI plant dynamics, the framework can be immediately extended to plant dynamics described by RNNs, or neural state space models provided in \citep{2018Kim}.
\end{remark}

\begin{figure*}[h]
	\centering
	\begin{minipage}{0.26\textwidth}
\includegraphics[width=1\textwidth]{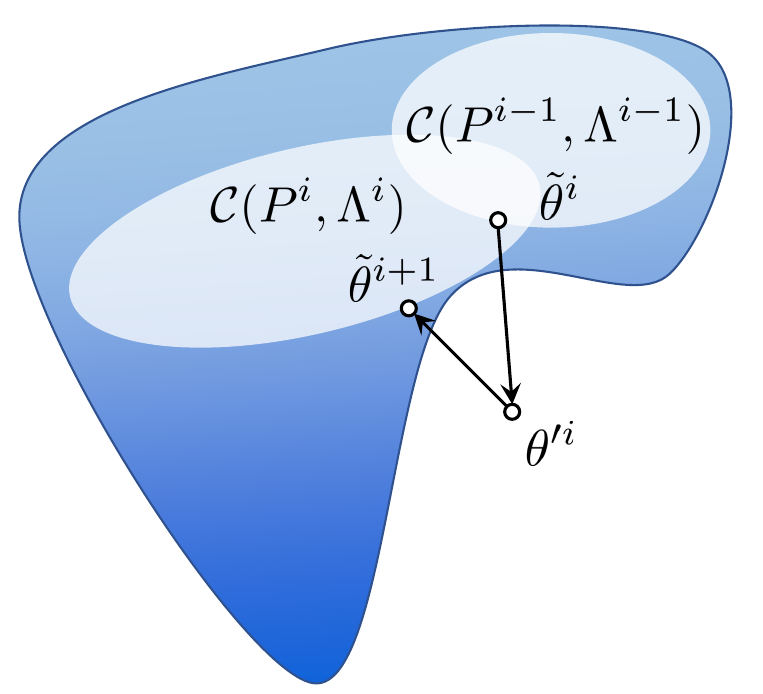}
    \captionof{figure}{Illustration of Algorithm~\ref{alg:alg1}. The set of all the stabilizing $\tilde{\theta}$ is given in blue.}
    \label{fig:projalg}  
	\end{minipage}
	\begin{minipage}{0.73\textwidth}
	\begin{subfigure}[b]{0.4\textwidth}
		\centering
		\includegraphics[width=1\textwidth]{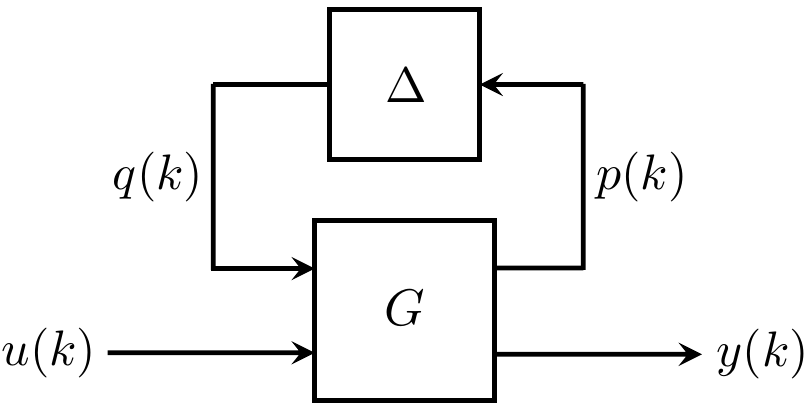}
		\caption{uncertain plant $F_u(G,\Delta)$}  
		\label{fig:Fu}
	\end{subfigure}
	\begin{subfigure}[b]{0.6\textwidth}  
		\centering 
		\includegraphics[width=1\textwidth]{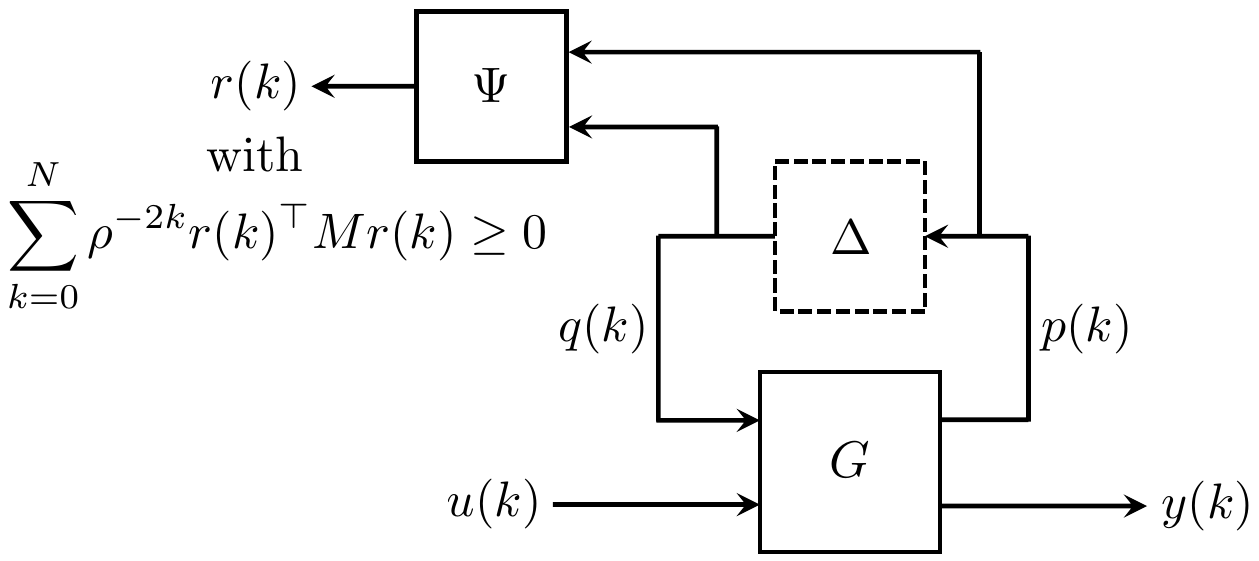}
		\caption{extended system of $G$ and $\Psi$ with IQCs}
    \label{fig:extended_sysm}   
	\end{subfigure}
	\caption{Uncertain plant and its corresponding constrained extended system }
	\end{minipage}
\end{figure*}


\subsection{Projected Policy Gradient}
Policy gradient methods \citep{sutton1999policy, williams1992simple} enjoy convergence to optimality under the tabular setting and achieve good empirical performance for more challenging problems. However, with little assumption about the problem setting, they do not offer any stability guarantee for the closed loop system. We propose the projected policy gradient method that enforces the stability of the interconnected system while the policy is dynamically explored and updated.

Policy gradient approximates the gradient with respect to the parameters of a stochastic controller using samples of trajectories via (\ref{eq:pg}) without any prior knowledge of plant parameters and the reward structures. Gradient ascent is then applied to refine the controller with the estimated gradients.
\small{
\begin{align}
\nabla_{\tilde{\theta}} R(\pi_{\tilde{\theta}}) &= \int_\mathcal{X} d^{\pi_{\tilde{\theta}}} (x) \int_{\mathcal{U}} \nabla_{\tilde{\theta}} \pi_{\tilde{\theta}} (u|x) Q^{\pi_{\tilde{\theta}}}(x, u) du dx \notag \\
&= \mathbb{E}_{\tilde{\theta}, x\sim d^\pi, u\sim \pi_{\tilde{\theta}}} [Q^\pi(x, u) \nabla_{\tilde{\theta}} \log \pi_{\tilde{\theta}} (u|x)]. \label{eq:pg}
\end{align}
}

In the above, $\tilde{\theta}$ represents the parameters of $\pi_{\tilde{\theta}}$. $R(\pi_{\tilde{\theta}})$ is the expected reward (negative cost) of the controller $\pi_{\tilde{\theta}}$. $d^\pi (x)$ is the distribution of states $x\in\mathcal{X}$ under $\pi_{\tilde{\theta}}$, where $\mathcal{X}$ is a set of states. $Q^\pi(x,u)$ is the reward-to-go after executing control $u\in\mathcal{U}$ at state $x$ under $\pi_{\tilde{\theta}}$, where $\mathcal{U}$ is a set of actions. 

Like any gradient method, policy gradient does not ensure the controller is in some specific set of preference (the set of stabilizing controller in our setting). To that end, a projection to the  stability set $\mathcal{C}(\bar{P}, \bar{\Lambda})$, $(Q_1, Q_2, \tilde{\theta}) \leftarrow \Pi_{\mathcal{C}(\bar{P}, \bar{\Lambda})}(\theta')$, is applied between gradient updates, where $\theta'$ is the updated parameter, 
 and the projection operator $\Pi_{\mathcal{C}(\bar{P}, \bar{\Lambda})}(\theta')$ is defined as the following convex program, 
\begin{align}
\Pi_{\mathcal{C}(\bar{P}, \bar{\Lambda})}(\theta') \in \arg\min_{\substack{Q_1, Q_2, \tilde{\theta}}}& \;\; {\left\|\begin{bmatrix}
Q_1 \\ Q_2
\end{bmatrix} - \begin{bmatrix}
\bar{P}^{-1} \\ \bar{\Lambda}^{-1}
\end{bmatrix}\right\|_F^2 + \|\tilde{\theta} - \theta' \|_F^2} \notag \\
\text{s.t. } &\;\; \text{LMI}(Q_1,Q_2,\tilde{\theta}, \bar{P}, \bar{\Lambda}). \label{eq:projection}
\end{align}
Through the recursively feasible projection step (\ie~the feasibility is inherited in subsequent steps, summarized in Theorem~\ref{thm:recursive_feas} in Appendix~\ref{app:recursive_feas}), we conclude with a projected policy gradient method to synthesize stabilizing RNN controllers as summarized in Algorithm~\ref{alg:alg1} and illustrated in Fig.~\ref{fig:projalg}. 

\begin{algorithm}[h]
	\caption{Projected Policy Gradient}
	\label{alg:alg1}
	\begin{algorithmic}[1]
		\Require{Matrices $P^0$ and $\Lambda^0$ s.t. $\mathcal{C}(P^0, \Lambda^0)$ is not empty, learning rate $\sigma > 0$.}
		\State $i \gets 0$
		\While{not converged}
		\State $\theta'^i \gets \tilde{\theta}^i + \sigma \nabla_{\tilde{\theta}} R(\pi_{\tilde{\theta}^i})$ \Comment{gradient step}
		\State $(Q_1^{i+1}, Q_2^{i+1},\tilde{\theta}^{i+1}) \gets \Pi_{\mathcal{C}(P^{i}, \Lambda^{i})}(\theta'^i)$ \Comment{proj. step}
		\State $P^{i+1} \gets (Q_1^{i+1})^{-1}$, \ $\Lambda^{i+1} \gets (Q_2^{i+1})^{-1}$
		\State $i \gets i + 1$
		\EndWhile
		\Ensure{$\pi_{\tilde{\theta}}$}
	\end{algorithmic}
\end{algorithm}


In the algorithm, the gradient step performs gradient ascent using the estimated gradient $\nabla_{\tilde{\theta}} R(\pi(\tilde{\theta}))$. The projection step projects the updated parameters $\theta'^i$ from the gradient step to the convex stability set $\mathcal{C}(P^i,\Lambda^i)$, where $P^i$ and $\Lambda^i$ are computed using $Q_1^i$ and $Q_2^i$ from the previous projection step. We choose $\Lambda^0 = I_{n_\phi}$, and construct $P^0$ based on the method in \citep{scherer1997multiobjective}.

\begin{remark} \label{rem:proj}
The projection step (\ref{eq:projection}) is a semi-definite program (SDP) involving $O((n_\xi + n_\phi) \times (n_\xi + n_\phi))$ variables. The complexity of interior point SDP solvers usually scales cubically with the number of variables, potentially bringing computational burden when $(n_\xi+n_\phi)$ is large. Luckily, most high dimensional problems admit low dimension structures \citep{Wright-Ma-2021} and such overhead is only paid at training without further operations at deployment.
\end{remark}

\begin{figure*}[ht]
  \centering
  \begin{minipage}{0.45\textwidth}
  \centering 
  \includegraphics[width=1\textwidth]{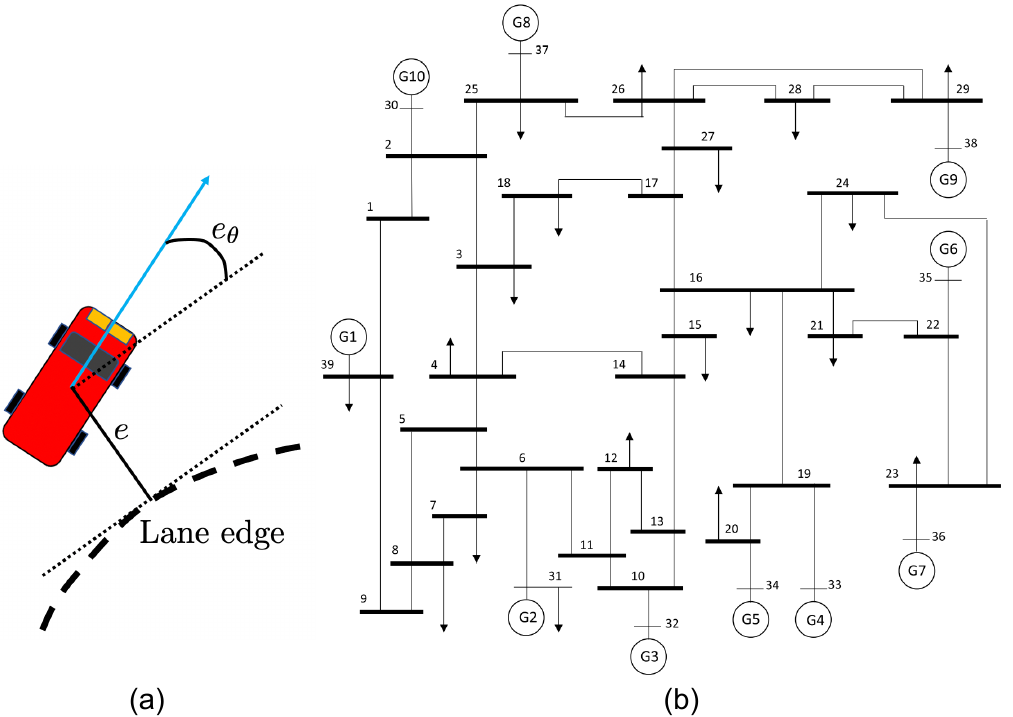}
  \caption{\small{(a) Vehicle \citep{alleyne1997comparison}; (b) Frequency Regulation on IEEE 39-bus New England Power System
  \citep{athay1979practical}}}
  \label{fig:envs}
  \end{minipage}
  \hspace{1cm}
  \begin{minipage}{0.32\textwidth}
  \includegraphics[width=0.95\textwidth]{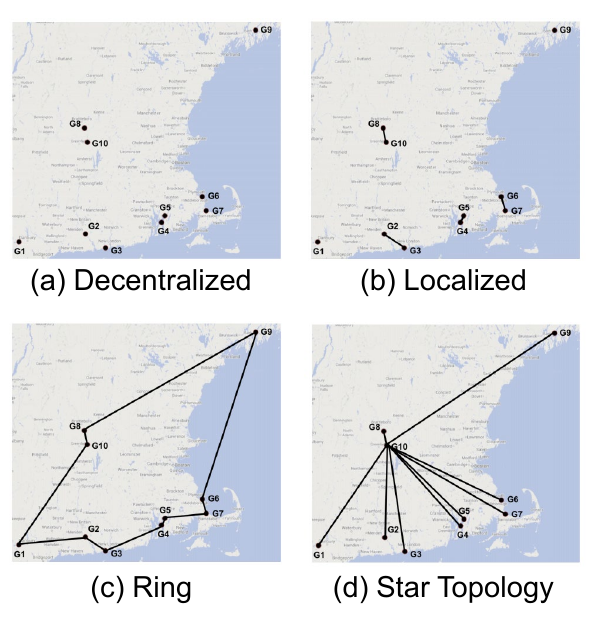}
  \caption{Four communication topologies for IEEE 39-bus power system \citep{fazelnia2016convex}.}
  \label{fig:ieee39}    
  \end{minipage}
\end{figure*} 

\section{Partially Observed Nonlinear Systems with
Uncertainty} \label{sec:IQC}
In the context of RL, we often need to handle systems with
nonlinear dynamics and/or unmodeled dynamics. Here we model such a nonlinear and uncertain plant $F_u(G,\Delta)$ (shown in Fig.~\ref{fig:Fu}) as an interconnection of the nominal plant $G$, and the perturbation $\Delta$ representing the nonlinear, and uncertain part of the system. Therefore, in this new problem setting, we only require system dynamics to be partially known, and we use $\Delta$ to cover the difference between the original real system dynamics, and partially known dynamics $G$. The plant $F_u(G, \Delta)$ is defined by the following equations:
\begin{align}
&G \left\{\begin{array}{ll}
x(k+1) &= A_G \ \ x(k)  + B_{G1} \ q(k) + B_{G2} \ u(k) \\
p(k) &= C_{G1} \ x(k) + D_{G1} \ q(k) \\
y(k) &= C_{G2} \ x(k) \end{array}\right. \notag\\
&\hspace{0.8cm} q(\cdot) \hspace{0.92cm} = \Delta(p(\cdot)) \label{eq:def_Delta}
\end{align}
where $x(k) \in \R^{n_G}$, $u(k) \in \R^{n_u}$, and $y(k) \in \R^{n_y}$ are the state, control input, and output of the nominal plant $G$, and $p(k)$ and $q(k)$ are the input and output of $\Delta$. The perturbation $\Delta: \ell_{2e}^{n_p} \rightarrow \ell_{2e}^{n_q}$ is a causal and bounded operator.


The perturbation $\Delta$ can represent various types of uncertainties and nonlinearities, including sector bounded nonlinearities, slope restricted nonlinearities, and unmodeled dynamics. Thus considering $\Delta$ extends our framework to the class of plants beyond LTI plants. The input-output relationship of $\Delta$ is characterized with an integral quadratic constraint (IQC) \citep{megretski1997system}, which consists of a filter $\Psi$ applied to the input $p$ and output $q$ of $\Delta$, and a constraint on the output $r$ of $\Psi$. The filter $\Psi$ is an LTI system with the zero initial condition $\psi(0) = 0_{n_{\psi \times 1}}$:
\begin{subequations}\label{eq:Psi_def}
\begin{align}
\begin{array}{ll}
    \psi(k+1) &= A_\psi \ \psi(k) + B_{\psi 1} \ p(k) + B_{\psi 2} \ q(k), \\
r(k) &= C_{\psi} \ \psi(k) + D_{\psi1} \ p(k) + D_{\psi 2} \ q(k).
\end{array}
\end{align}
\end{subequations}
To enforce exponential stability of the feedback system, we characterize $\Delta$ using the time-domain $\rho$-hard IQC, which is introduced in \citep{lessard2016analysis}, and its definition is also provided below. 
\begin{definition}
  Let $\Psi$ be an LTI system defined in \eqref{eq:Psi_def}, and $M \in \mathbb{S}^{n_r}$. Suppose $0 \leq \rho \leq 1$. A bounded, causal operator $\Delta : \ell_{2e}^{n_p} \rightarrow \ell_{2e}^{n_q}$ satisfies the time-domain $\rho$-hard IQC defined by $\Psi$, $M$, and $\rho$, if the following condition holds for all $p \in \ell_{2e}^{n_p}$, $q = \Delta(p)$, and $N \ge 0$
  \begin{align}
      \sum_{k=0}^N \rho^{-2k} r(k)^\top M r(k) \ge 0. \label{eq:def_IQC}
  \end{align}
  where $r$ is the output of $\Psi$ driven by inputs $(p, q)$.
\end{definition}
\begin{remark}
For a particular perturbation $\Delta$, there is typically a class of
valid $\rho$-hard IQCs defined by a fixed filter $\Psi$ and a
matrix $M$ drawn from a convex set $\mathcal{M}$. Thus, in the stability condition derived later, $M \in \mathcal{M}$ will also be treated as a decision variable. A library of frequency-domain $\rho$-IQCs is provided in \citep{boczar2017exponential} for various types of perturbations. As shown in \citep{schwenkel2021model}, a general class of frequency-domain $\rho$-IQCs can be translated into time-domain $\rho$-hard IQC by a multiplier factorization. 
\end{remark}

When deriving the stability condition, the perturbation $\Delta$ will be replaced by the time-domain $\rho$-hard IQC \eqref{eq:def_IQC} that describes it, and the associated filter $\Psi$, as shown in Fig.~\ref{fig:extended_sysm}. Therefore, the stabilizing controller will be designed for the extended system (an interconnection of $G$ and $\Psi$) subject to IQCs, instead of the original $F_u(G,\Delta)$. This controller will also be able to stabilize the original $F_u(G,\Delta)$.
Define the extended state as $x_e = [x^\top, \ \psi^\top]^\top$, and the dynamics of the extended system are given in Appendix~\ref{app:dyn_extendsys}. Define $\zeta = [x_e^\top, \ \ \xi^\top]^\top$ to gather the states of the extended system and the controller. The feedback system of the extended system and the controller has the dynamics
\begin{align}
\begin{array}{ll}
    \zeta(k+1) &= \mathcal{A} \ \zeta(k) + \mathcal{B}_1 \ q(k) + \ \mathcal{B}_2 \ z(k), \\
v(k) &= \mathcal{C}_1 \ \zeta(k) + \mathcal{D}_1 \ q(k) + \mathcal{D}_2 \ z(k), \\
r(k) &= \mathcal{C}_2 \ \zeta(k) + \mathcal{D}_3 \ q(k) + \mathcal{D}_4 \ z(k),
\end{array}
\end{align}
where 
\begin{align}
\begin{array}{ll}
&\hspace{0.13cm}\mathcal{A} = \smat{A_e + B_{e2}\tilde{D}_{k2}C_{e2} & B_{e2}\tilde{C}_{k1} \\ \tilde{B}_{k2} C_{e2} & \tilde{A}_k}, \hspace{0.85cm}\mathcal{B}_1 =\smat{B_{e1}  \\ 0_{n_\xi \times n_q} }, \\
&\hspace{0.03cm}\mathcal{B}_2 =\smat{ B_{e2} D_{k1}\frac{B_\phi-A_\phi}{2}  \\ B_{k1}\frac{B_\phi-A_\phi}{2}}, \hspace{1.85cm} \mathcal{C}_1 = \smat{D_{k 3} C_{e2} & C_{k2}}, \nonumber \\
&\mathcal{D}_1 = 0_{n_\phi \times n_q}, \hspace{0.65cm} \mathcal{D}_2 = 0_{n_\phi \times n_\phi}, \hspace{0.6cm} \mathcal{C}_2 = \smat{C_{e1}  & 0_{n_r \times n_\xi}}, \\ 
&\mathcal{D}_3 = D_{e1}, \hspace{1.13cm} \mathcal{D}_4 = 0_{n_r \times n_\phi}, \nonumber
\end{array}
\end{align}
and the state space matrices $(A_e, B_{e1}, ..., D_{e1})$ of the extended system are defined in Appendix~\ref{app:dyn_extendsys}. The next theorem merges the QC for $\tilde{\phi}$ and the time-domain $\rho$-hard IQC for $\Delta$ with the Lyapunov theorem to derive the exponential stability condition for the uncertain feedback system.  


\begin{figure*}[h!]
  \centering
  \includegraphics[width=0.85\textwidth]{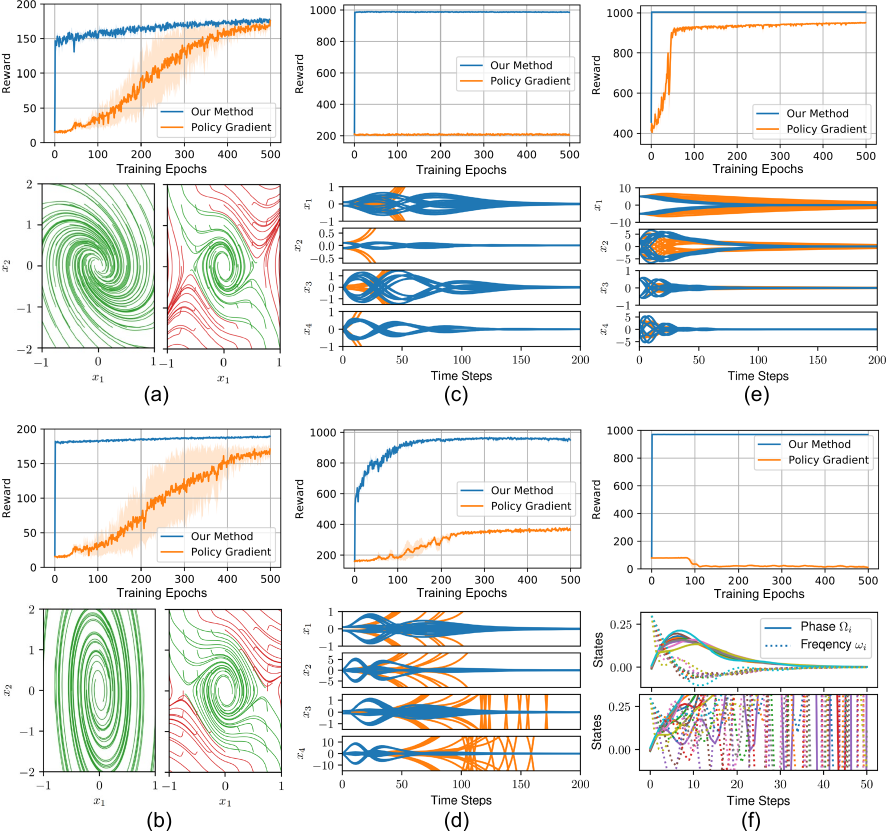}
  \caption{(a) Inverted Pendulum (linear); (b) Inverted Pendulum (nonlinear); (c) Cartpole;  (d) Pendubot; (e) Vehicle lateral control; (f) IEEE 39-bus New England Power System frequency regulation. The error bars of reward plots characterize standard deviation across 3 runs with different seeds. For (a) and (b), the left figures are from our method and right figures from policy gradient. Converging trajectories are rendered in green while diverging ones in red. For (c), (d), (e), trajectories from our method are given in blue while those from policy gradient are in orange. For (f), top figure is given by our method and bottom one by policy gradient.} 
  \label{fig:experiments}    
\end{figure*}

\begin{theorem}\label{thm:lyap_IQC}
Consider the feedback system of uncertain plant $F_u(G,\Delta)$, and RNN controller $\pi_{\tilde{\theta}}$. Assume $\Delta$ satisfies the time-domain $\rho$-hard IQC defined by $\Psi$, $\mathcal{M}$, and $\rho$, with $0 \leq \rho \leq 1$. Given $\bar{P} \in \R^{n_\zeta \times n_\zeta}$ and $\bar{\Lambda} \in \R^{n_\phi \times n_\phi}$. If there exist matrices $Q_1 \in \mathbb{S}_{++}^{n_\zeta}$, $Q_2 \in \mathbb{D}_{++}^{n_\phi}$, $M \in \mathcal{M}$, and parameters $\tilde{\theta}$  such that the following condition holds
\begin{align}\label{eq:convex_robust_stab}
    \bmat{R^\top \Gamma R ~~~~~~~~~ & \begin{bmatrix}\mathcal{A} & \mathcal{B}_1 & \mathcal{B}_2 \\ \mathcal{C}_1 & \mathcal{D}_1 & \mathcal{D}_2 \end{bmatrix}^\top \\ \begin{bmatrix}\mathcal{A} & \mathcal{B}_1 & \mathcal{B}_2 \\ \mathcal{C}_1 & \mathcal{D}_1 & \mathcal{D}_2 \end{bmatrix} & \begin{bmatrix} Q_1 & 0 \\ 0& Q_2 \end{bmatrix} } \succeq 0,
\end{align}
where $\Gamma = \diag(\rho^2 (2\bar{P} - \bar{P}^\top Q_1 \bar{P}), 2\bar{\Lambda} - \bar{\Lambda}^\top Q_2 \bar{\Lambda}, -M)$ is a block diagonal matrix and $R = \left[\begin{smallmatrix}I & 0 & 0 \\ 0 & 0 & I \\ \mathcal{C}_{2} & \mathcal{D}_{3} & \mathcal{D}_{4} \end{smallmatrix}\right]$.
Then for any $x(0)$, we have $\|x(k)\| \leq \sqrt{\text{cond}(P)}\rho^k \|x(0)\|$ for all $k$, where $P := Q_1^{-1}$, i.e., the feedback system is exponentially stable with rate $\rho$. 
\end{theorem}
The complete proof is provided in Appendix~\ref{app:pf_lyapIQC}. This LMI~\eqref{eq:convex_robust_stab} is jointly convex in $\tilde{\theta}, Q_1, Q_2$ and $M$ for any given $\bar{P}$ and $\bar{\Lambda}$. Based on this LMI, we define the convex robust stability set $\mathcal{C}_R(\bar{P},\bar{\Lambda})$:
\begin{align*}
    \mathcal{C}_R(\bar{P},\bar{\Lambda}) := \left\{ \tilde{\theta}: \ \exists \ Q_1 \in \mathbb{S}_{++}^{n_\zeta}, \ Q_2 \in \mathbb{D}_{++}^{n_\phi}, \ M \in \mathcal{M}, \ \text{s.t.} \ \eqref{eq:convex_robust_stab} \right\}.
\end{align*}

Any parameter $\tilde{\theta}$ drawn from $\mathcal{C}_R(\bar{P}, \bar{\Lambda})$ ensures the exponential stability of the feedback system of $F_u(G,\Delta)$ and $\pi_{\tilde{\theta}}$, and this convex robust stability set can be used in the projection step. 
\begin{remark}
If we only require the feedback system to be stable ($\rho = 1$ in \eqref{eq:convex_robust_stab}), a more general class of IQCs, the time-domain hard IQCs \citep{megretski1997system}, can be used to describe $\Delta$. 
\end{remark}

\section{Numerical Experiment} \label{sec:numerical} 
To compare our method against regular RNN controller trained without projection, we consider 6 different tasks involving control of partially observed dynamical systems, including a linearized inverted pendulum and its nonlinear variant, a cartpole, vehicle lateral dynamics, a pendubot, and a high dimensional power system. 
Fig.~\ref{fig:envs} gives a demonstrative visualization of tasks including vehicle lateral control and IEEE 39-bus power system frequency regulation, whose communication topologies are shown in Fig.~\ref{fig:ieee39}. Experimental settings and tasks definitions are detailed in Appendix~\ref{app:more_experiments}. \blfootnote{Code available at \url{https://github.com/beeperman/IQCRNN}}

The experimental results including rewards and sample trajectories at convergence are reported in Fig.~\ref{fig:experiments}. In all experiments, our method achieves high reward after the first few projection steps that ensures stability, greatly outperforming the regular method which suffers from instability even after converging. For pendubot and inverted pendulum tasks, our method keeps perfecting the performance after the first projection steps which already give high performance. While for cartpole, vehicle lateral control, and power system frequency regulation tasks, our method converges to optimal performance in one step. Our method gives converging trajectories for all tasks and achieves faster converging trajectories on the vehicle lateral control task. 
In comparison, policy gradient has been greatly impacted by the partial observability and converges to sub-optimal performance in cartpole, pendubot, and power system frequency regulation tasks and requires more steps to achieve optimal performance in inverted pendulum and vehicle lateral control tasks. Without stability guarantee, policy gradient fails to ensure converging trajectories from some initial conditions for all tasks excluding vehicle lateral control which is open-loop stable.



\section{Conclusion}
In this work, we present a method to synthesize stabilizing RNN controllers, which ensures the stability of the feedback systems both during learning and control process. 
We develop a convex set of stabilizing RNN parameters for nonlinear and partially observed systems. A novel projected policy gradient method is developed to synthesize a controller while enforcing stability by recursively projecting the parameters of the RNN controller to the convex set. By evaluating on a variety of control tasks, we demonstrate that our method learns stabilizing controllers with fewer samples, faster converging trajectories, and higher final performance than policy gradient.
Future directions include extensions to implicit models \citep{bai2019deep, elghaoui2020}, and model-free RL  \citep{jiang2012computational}.






\bibliography{ref.bib}

\onecolumn
\appendix

{\centering{{\LARGE\bfseries Supplementary material}}}

\section{Proofs and Illustrations} \label{app:proofs}
\subsection{Derivation for the Transformed Plant $\tilde{P}_\pi$}\label{app:tranplant}
The input to $P_\pi$ is transformed by the following equation:
\begin{align}
    w(k) = \frac{B_\phi-A_\phi}{2}z(k) + \frac{A_\phi+B_\phi}{2}v(k).\label{eq:tran_input}
\end{align}
Substituting the expression of $v(k)$ from \eqref{eq:PK_def} into \eqref{eq:tran_input} yields
\begin{align}
    w(k) = \frac{B_\phi-A_\phi}{2}z(k) + \frac{A_\phi+B_\phi}{2}C_{K2}\ \xi(k) + \frac{A_\phi+B_\phi}{2}D_{K3}\ y(k).\label{eq:tran_input_final}
\end{align}
Finally, the transformed plant $\tilde{P}_\pi$ can be obtained by substituting \eqref{eq:tran_input_final} into \eqref{eq:PK_def}.

\subsection{Proof of Theorem~\ref{thm:lyap_nominal}}\label{app:pf_nominal_stab}
\begin{proof}
Assume there exist $Q_1 \in \mathbb{S}_{++}^{n_\zeta}$, $Q_2 \in \mathbb{D}_{++}^{n_\phi}$, and $\tilde{\theta}$, such that \eqref{eq:lyap_convex} holds. It follows from Schur complements that \eqref{eq:lyap_convex} is equivalent to 
\begin{align}
     \begin{bmatrix}\mathcal{A} & \mathcal{B} \\ \mathcal{C} & \mathcal{D} \end{bmatrix}^\top \begin{bmatrix}Q_1^{-1} & 0 \\ 0 & Q_2^{-1} \end{bmatrix} \begin{bmatrix}\mathcal{A} & \mathcal{B} \\ \mathcal{C} & \mathcal{D} \end{bmatrix} - \begin{bmatrix}\rho^2 (2\bar{P} - \bar{P}^\top Q_1 \bar{P}) & 0 \\ 0 & 2\bar{\Lambda} - \bar{\Lambda}^\top Q_2 \bar{\Lambda} \end{bmatrix} \preceq 0. \label{eq:lemma1_lyap1}
\end{align}
It follows from the inequalities $\bar{P}^\top Q_1 \bar{P} - 2\bar{P} \succeq -Q_1^{-1}$ and $\bar{\Lambda}^\top Q_2 \bar{\Lambda} - 2\bar{\Lambda} \succeq -Q_2^{-1}$ for any $\bar{P} \in \R^{n_\zeta \times n_\zeta}$ and $\bar{\Lambda} \in \R^{n_\phi \times n_\phi}$ \citep{tobenkin2017convex, revay2020convex} that \eqref{eq:lemma1_lyap1} implies
\begin{align}\label{eq:lyap_rearr}
 \begin{bmatrix}\mathcal{A} & \mathcal{B} \\ \mathcal{C} & \mathcal{D} \end{bmatrix}^\top \begin{bmatrix}Q_1^{-1} & 0 \\ 0 & Q_2^{-1} \end{bmatrix} \begin{bmatrix}\mathcal{A} & \mathcal{B} \\ \mathcal{C} & \mathcal{D} \end{bmatrix} -   \begin{bmatrix}\rho^2 Q_1^{-1} & 0 \\ 0 & Q_2^{-1} \end{bmatrix} \preceq 0.
\end{align}
Defining $P = Q_1^{-1}$, and $\Lambda = Q_2^{-1}$, and rearranging \eqref{eq:lyap_rearr}, we have that $P, \Lambda$, and $\tilde{\theta}$ satisfy the following condition
\begin{align}\label{eq:lyap_cond}
\begin{bmatrix}\mathcal{A}^\top P \mathcal{A}- \rho^2 P & \mathcal{A}^\top P \mathcal{B} \\ \mathcal{B}^\top P \mathcal{A} & \mathcal{B}^\top P \mathcal{B} \end{bmatrix}  +  \begin{bmatrix}\mathcal{C} & \mathcal{D} \\ 0 & I \end{bmatrix}^\top \begin{bmatrix}\Lambda & 0 \\ 0 & -\Lambda \end{bmatrix}\begin{bmatrix}\mathcal{C} & \mathcal{D} \\ 0 & I \end{bmatrix}  \preceq 0,
\end{align}
Define the Lyapunov function $V(\zeta):= \zeta^\top P \zeta$. Multiplying \eqref{eq:lyap_cond} on the left and right by $[\zeta(k)^\top, \ z(k)^\top]$ and its transpose yields
\begin{align}
    V(\zeta(k+1)) - \rho^2 V(\zeta(k)) + \bmat{v(k) \\ z(k)}^\top \bmat{\Lambda & 0 \\ 0 & -\Lambda} \bmat{v(k) \\ z(k)} \leq 0.\label{eq:pf1_lyap}
\end{align}
It follows from $\tilde{\phi} \in$ sector $[-1_{n_\phi \times 1}, 1_{n_\phi \times 1}]$ that the last term in \eqref{eq:pf1_lyap} is nonnegative, and thus $V(\zeta(k+1)) \leq \rho^2 V(\zeta(k))$. Iterate it down to $k=0$, we have $V(\zeta(k)) \leq \rho^{2k}V(\zeta(0))$, which implies $\|\zeta(k)\| \leq \sqrt{\text{cond}(P)} \rho^k \|\zeta(0)\|$. Recall $\xi(0) = 0$. Therefore $$\|x(k)\| \leq \|\zeta(k)\| 
            \leq \sqrt{\text{cond}(P)} \rho^k \|\zeta(0)\| 
            =\sqrt{\text{cond}(P)} \rho^k \|x(0)\|, $$
and this completes the proof. 
\end{proof}



\subsection{Recursive Feasibility of the Projection Steps}\label{app:recursive_feas}
\begin{theorem} [Recursive Feasibility]\label{thm:recursive_feas}
If LMI$(Q_1, Q_2, \tilde{\theta}, P^{i-1}, \Lambda^{i-1})$ is feasible (\ie ~ LMI$(Q_1, Q_2, \tilde{\theta}, P^{i-1}, \Lambda^{i-1})$ holds for some $Q_1$, $Q_2$, and $\tilde{\theta}$), then LMI$(Q_1, Q_2, \tilde{\theta}, P^{i}, \Lambda^{i})$ is also feasible, where $P^i = (Q_1^i)^{-1}$ and $\Lambda^i = (Q_2^i)^{-1}$ are from the $i$-th step of projected policy gradient, for $i = 1, 2, \ldots$
\end{theorem}
\begin{proof}
The main idea is to show that $(Q_1^i, Q_2^i, \tilde{\theta}^i)$ is already a feasible point for LMI$(Q_1, Q_2, \tilde{\theta}, P^{i}, \Lambda^{i})$. Since LMI$(Q_1, Q_2, \tilde{\theta}, P^{i-1}, \Lambda^{i-1})$ is feasible, at optimum of the projection step, we obtain the minimizer $(Q_1^i, Q_2^i, \tilde{\theta}^i)$ and LMI$(Q_1^i, Q_2^i, \tilde{\theta}^i, P^{i-1}, \Lambda^{i-1})$ holds. It follows from inequalities $2P^{i-1} - P^{i-1}{}^\top Q_1^i P^{i-1} \preceq (Q_1^i)^{-1}$ and $2\Lambda^{i-1} - \Lambda^{i-1}{}^\top Q_2^i \Lambda^{i-1} \preceq (Q_2^i)^{-1}$ that 
\begin{align}
         \begin{bmatrix}\rho^2(Q_1^i)^{-1} & 0 & \mathcal{A}^\top & \mathcal{C}^\top \\ 0 & (Q_2^i)^{-1} & \mathcal{B}^\top& \mathcal{D}^\top \\ \mathcal{A} & \mathcal{B} & Q_1^i & 0\\ \mathcal{C} & \mathcal{D} & 0 &Q_2^i \end{bmatrix} 
         \succeq 0,
\end{align}
which renders LMI$(Q_1^i, Q_2^i, \tilde{\theta}^i, P^{i}, \Lambda^{i})$ true at a feasible point of $(Q_1^i, Q_2^i, \tilde{\theta}^i)$.
\end{proof}

\subsection{Dynamics of the Extended System}\label{app:dyn_extendsys}
The extended system (shwon in Fig.~\ref{fig:extended_sysm}) is of the form
\begin{subequations}
\begin{align}
    x_e(k+1) &= A_e \ x_e(k) + B_{e1} \ q(k) + B_{e2} \ u(k) \\
r(k) &= C_{e1} \ x_e(k) + D_{e1} \ q(k)  \\
y(k) &= C_{e2} \ x_e(k)
\end{align}
\end{subequations}
where 
\begin{subequations}
\begin{align}
    A_{e} &= \begin{bmatrix}A_G & 0 \\ B_{\psi 1} C_{G1} & A_\psi\end{bmatrix}, \ B_{e1} =\begin{bmatrix}B_{G1} \\B_{\psi 1}D_{G1} + B_{\psi2}  \end{bmatrix}, \ B_{e2} =\begin{bmatrix} B_{G2} \\ 0 \end{bmatrix}\nonumber \\
C_{e1} &= \begin{bmatrix}D_{\psi 1}C_{G1} & C_\psi\end{bmatrix}, \ \ \ D_{e1} = \begin{bmatrix}D_{\psi 1} D_{G1} + D_{\psi 2} \end{bmatrix}, \ \ C_{e2} = \begin{bmatrix}C_{G2} & 0 \end{bmatrix}.
\end{align}
\end{subequations}

\subsection{Proof of Theorem~\ref{thm:lyap_IQC}}\label{app:pf_lyapIQC}
\begin{proof}
Assume there exist $Q_1 \in \mathbb{S}_{++}^{n_\zeta}$, $Q_2 \in \mathbb{D}_{++}^{n_\phi}$, $M \in \mathcal{M}$, and $\tilde{\theta}$ such that \eqref{eq:convex_robust_stab} holds. It follows from Schur complements that \eqref{eq:convex_robust_stab} is equivalent to 
\begin{align}
    \scriptsize{\begin{bmatrix}\mathcal{A} & \mathcal{B}_1 & \mathcal{B}_2 \\ \mathcal{C}_1 & \mathcal{D}_1 & \mathcal{D}_2 \end{bmatrix}^\top \begin{bmatrix} Q_1^{-1} & 0 \\ 0& Q_2^{-1} \end{bmatrix} \begin{bmatrix}\mathcal{A} & \mathcal{B}_1 & \mathcal{B}_2 \\ \mathcal{C}_1 & \mathcal{D}_1 & \mathcal{D}_2 \end{bmatrix}  - R^\top \begin{bmatrix} \rho^2(2\bar{P} - \bar{P}^\top Q_1 \bar{P}) & 0 &0\\ 0& 2\bar{\Lambda}-\bar{\Lambda}^\top Q_2 \bar{\Lambda} & 0 \\ 0&0& -M \end{bmatrix} R} \preceq 0.\label{pf_lemma2_lyap1}
\end{align}
By inequalities $\bar{P}^\top Q_1 \bar{P} - 2\bar{P} \succeq -Q_1^{-1}$ and $\bar{\Lambda}^\top Q_2 \bar{\Lambda} - 2\bar{\Lambda} \succeq -Q_2^{-1}$ for any $\bar{P} \in \R^{n_\zeta \times n_\zeta}$ and $\bar{\Lambda} \in \R^{n_\phi \times n_\phi}$, we have that \eqref{pf_lemma2_lyap1} implies 
\begin{align}
    \begin{bmatrix}\mathcal{A} & \mathcal{B}_1 & \mathcal{B}_2 \\ \mathcal{C}_1 & \mathcal{D}_1 & \mathcal{D}_2 \end{bmatrix}^\top \begin{bmatrix} Q_1^{-1} & 0 \\ 0& Q_2^{-1} \end{bmatrix} \begin{bmatrix}\mathcal{A} & \mathcal{B}_1 & \mathcal{B}_2 \\ \mathcal{C}_1 & \mathcal{D}_1 & \mathcal{D}_2 \end{bmatrix}  - R^\top \begin{bmatrix} \rho^2 Q_1^{-1} & 0 &0  \\ 0& Q_2^{-1} & 0 \\ 0&0& -M\end{bmatrix} R  \preceq 0. \label{eq:lyap_intermediate}
\end{align}
Defining $P = Q_1^{-1}$ and $\Lambda = Q_2^{-1}$, and rearranging  \eqref{eq:lyap_intermediate}, we have  $P, \Lambda, M$, and $\tilde{\theta}$ satisfy the following condition
\begin{align}\label{eq:lyapIQC_cond}
    &\bmat{\mathcal{A} & \mathcal{B}_1 & \mathcal{B}_2 \\ I & 0 & 0}^\top \bmat{P & 0 \\ 0 & -\rho^2 P}\bmat{\mathcal{A} & \mathcal{B}_1 & \mathcal{B}_2 \\ I & 0 & 0} + \bmat{\mathcal{C}_1 & \mathcal{D}_1 & \mathcal{D}_2 \\ 0 & 0 & I}^\top \bmat{\Lambda & 0 \\ 0 & -\Lambda} \bmat{\mathcal{C}_1 & \mathcal{D}_1 & \mathcal{D}_2 \\ 0 & 0 & I} \nonumber \\
    & \hspace{0.6cm}+ \bmat{\mathcal{C}_2 & \mathcal{D}_3 & \mathcal{D}_4}^\top M \bmat{\mathcal{C}_2 & \mathcal{D}_3 & \mathcal{D}_4} \preceq 0.
\end{align}
Define the Lyapunov function $V(\zeta):= \zeta^\top P \zeta$. Multiplying \eqref{eq:lyapIQC_cond} on the left and right by $\bmat{\zeta(k)^\top, \ q(k)^\top, \ z(k)^\top}$ and its transpose yields
\begin{align}
    V(\zeta(k+1)) - \rho^2 V(\zeta(k)) + \bmat{v(k) \\ z(k)}^\top \bmat{\Lambda & 0 \\ 0 & -\Lambda} \bmat{v(k) \\ z(k)} + r(k)^\top M r(k) \leq 0. \label{eq:pf2_lyap1}
\end{align}
It follows from $\tilde{\phi} \in$ sector $[-1_{n_\phi \times 1}, 1_{n_\phi \times 1}]$ that the third term is nonnegative. This yields 
\begin{align}
    V(\zeta(k+1)) - \rho^2 V(\zeta(k)) + r(k)^\top M r(k) \leq 0. \label{eq:pf2_lyap2}.
\end{align}
Multiply \eqref{eq:pf2_lyap2} by $\rho^{-2k}$ for each $k$ and sum over $k$ to obtain
\begin{align}
    \rho^{-2(k-1)}V(\zeta(k)) - \rho^2 V(\zeta(0)) + \sum_{t=0}^{k-1} \rho^{-2t} r(t) M r(t) \leq 0.
\end{align}
By the assumption that $\Delta$ satisfies the $\rho$-hard IQC, the last term is nonnegative, and thus $V(\zeta(k)) \leq \rho^{2k} V(\zeta(0))$ for all $k$, which implies $\|\zeta(k)\| \leq \sqrt{\text{cond}(P)} \rho^k \|\zeta(0)\|$. Recall $\xi(0) = 0_{n_\xi \times 1}$ and $\psi(0) = 0_{n_{\psi} \times 1}$. Therefore
    $$\|x(k)\| \leq \|\zeta(k)\|
            \leq \sqrt{\text{cond}(P)} \rho^k \|\zeta(0)\| 
            =\sqrt{\text{cond}(P)} \rho^k \|x(0)\|, $$
and this completes the proof.
\end{proof}

\section{More on Experiments}
\label{app:more_experiments}

In this section, we give detailed information about the experiments. All experiments are conducted on a custom built machine with 36-core Intel Broadwell Xeon CPUs with 64 GB of RAM and are terminated in tens of minutes. In all tasks, both our method and policy gradient are trained to convergence and are capped at 1000 epochs. For each epoch, the gradient is estimated from a batch of 6000 step data sampled from the controller interacting with the plant and is applied to update the parameters. The trajectory length is capped at $200$. The average reward from the sample of trajectories from the 6000 steps are reported. We show 500 epoch plots in Fig.~\ref{fig:experiments} for clearer capture of the convergence process. The experiments and models are coded in Python \citep{python} with Tensorflow \citep{tensorflow2015-whitepaper}, CVXPY \citep{diamond2016cvxpy} and MOSEK\footnote{\url{www.mosek.com}}. We choose the learning rate of 1e-3, picked from grid search from 1e-1, 1e-2, 1e-3, 1e-4 to give best reward at convergence for policy gradient on the inverted pendulum task. We use ADAM optimizer \citep{kingma2014adam} and gradient clipping to a maximum magnitude of $10$ for each parameter. We use $\tanh$ activation for all our neural network controllers. The implementation of policy gradient is consistent with \citet{cs285}. And our method adds the projection updates on the same code base. The initial guess for the Lyapunov matrix $P^0$ of Algorithm~\ref{alg:alg1} is constructed using the method introduced in \citep{scherer1997multiobjective}, which computes the Lyapunov matrix for output feedback control problem through LMIs. The initial guess for $\Lambda^0$ is an identity matrix.


\subsection{Inverted Pendulum} \label{app:invetpen}

We consider both a linearized inverted pendulum system and a full nonlinear version whose dynamics are given below. 
Both examples have two states $x_1$, $x_2$, representing the angular position (rad) and velocity (rad/s). 
Only the plant output, $y = x_1$, is observed. Our methods as described in Section~\ref{sec:LTI} and \ref{sec:IQC} are applied for the linear, and nonlinear variants, with the goal of balancing the inverted pendulum around the upright position. 

Consider an inverted pendulum with mass $m = 0.15$ kg, length $l = 0.5$ m, and
friction coefficient $\mu = 0.5$ Nms$/$rad. The discretized and linearized dynamics are:
	\begin{align}
	\bmat{x_1(k+1) \\ x_2(k+1)} &= \bmat{1 & \delta \\ \frac{g \delta}{l} & 1-\frac{\mu \delta}{ml^2}}\bmat{x_1(k) \\ x_2(k)} + \bmat{0\\ \frac{\delta}{ml^2}}u(k), \nonumber \\
	y(k) & = \bmat{1 & 0}\bmat{x_1(k) \\ x_2(k)},
	\end{align}
 where $u$ is the control
input (1/10 Nm), and $\delta = 0.02$~s is the sampling time.

The discretized nonlinear dynamics of the inverted pendulum are
	\begin{align*}
	\bmat{x_1(k+1) \\ x_2(k+1)} &= \bmat{1 & \delta \\ \frac{g \delta}{l} & 1-\frac{\mu \delta}{ml^2}}\bmat{x_1(k) \\ x_2(k)} + \bmat{0 \\ \frac{-g \delta}{l}} q(k) + \bmat{0\\ \frac{\delta}{ml^2}}u(k), \\
	y(k) & = \bmat{1 & 0}\bmat{x_1(k) \\ x_2(k)},\\
	q(k) &= \Delta(x_1(k)) = x_1(k) - \sin(x_1(k)).
	\end{align*}
The nonlinearity $\Delta(x_1) = x_1 - \sin(x_1)$ lies in the sector $[0, 0.41]$ for $x_1 \in [-1.4 \ \text{rad}, 1.4 \ \text{rad}]$. The time domian $\rho$-hard IQC for describing the nonlinearity $\Delta$ is defined by the static filter $\Psi = \smat{0.41 & -1 \\ 0 & 1}$, the matrix $M = \smat{0 & \lambda \\ \lambda & 0}$ for all $\lambda \ge 0$, and any $\rho \ge 0$.

At training, we pick $\rho = 1$ for both tasks. The observation of output is limited to $[-0.15, 0.15]$ and is normalized (\ie~ devided by $0.15$) before feeding into the controller. The trajectory terminates when the limit is violated or the length arrives 200. The hyperparameters are set to $n_\phi = 16, n_\xi = 16$. The following reward is used,
$$
R = \sum_{k=0}^T \left[1.0 - 100 x_1 (k)^2 - 10 x_2(k)^2 + 100 u(k)^2 \right],
$$
where $(x_1(k), x_2(k))$ is state and $u(k)$ is control at step $k$.

\subsection{Cartpole}\label{app:cartpole}

A linearized  cartpole system is considered, the goal of which is to balance the pendulum around the upright position while keeping the position of the cart close to the origin.  The discretized and linearized dynamics of the cartpole are
\begin{align*}
    \bmat{x_1(k+1) \\ x_2(k+1) \\ x_3(k+1)\\ x_4(k+1)} &= \bmat{1 & -0.001 & 0.02 & 0 \\ 0 & 1.005 & 0 & 0.02 \\ 0 & -0.079 & 1 & -0.001 \\ 0 & 0.55 & 0 & 1.005}\bmat{x_1(k) \\ x_2(k) \\ x_3(k)\\ x_4(k)} + \bmat{0 \\ 0 \\ 0.04 \\ -0.04}u(k), \\
    y(k) & = \bmat{1 & 0 &  0 & 0 \\ 0 & 1 & 0 & 0}\bmat{x_1(k) \\ x_2(k) \\ x_3(k) \\ x_4(k)},
\end{align*}
where $x_1$ and $x_2$ represent the cart position (m) and the angular position (rad) of the pendulum,  $x_3$ and $x_4$ are the corresponding cart velocities (m/s) and angular velocity (rad/s) of the pendulum, $u$ is the horizontal force (N) exerting on the cart, and $y$ is the plant output.

At training, we pick $\rho = 0.98$. The observation of output is limited to $x_1: [-1, 1], x_2: [-\pi/2, \pi/2]$ and is normalized before feeding into the controller. The trajectory terminates when the limit is violated or the length arrives 200. The hyperparameters are set to $n_\phi = 16, n_\xi = 16$. The following reward is used,
$$
R = \sum_{k=0}^T \left[5.0 - x_1(k)^2 - x_2(k)^2 - 0.04 x_{3}(k)^2 - 0.1 x_4(k)^2 - 0.2 u(k)^2 \right],
$$
where $(x_1(k), x_2(k), x_3(k), x_4(k))$ is state and $u(k)$ is control at step $k$.

\subsection{Pendubot} \label{app:pendubot}

A linearized pendubot is considered. The main goal is to balance the pendubot around the upright position.  The linearized and discretiezed dynamics (sampling time $\delta = 0.01$~s) of the pendubot are
\begin{align*}
    \bmat{x_1(k+1) \\ x_2(k+1) \\ x_3(k+1)\\ x_4(k+1)} &= \bmat{1 & 0.01 & 0 & 0 \\ 0.6738 & 1 & -0.2483 & 0 \\ 0 & 0 & 1 & 0.01 \\ -0.6953 & 0 & 1.0532 & 1}\bmat{x_1(k) \\ x_2(k) \\ x_3(k)\\ x_4(k)} + \bmat{0 \\ 0.4487 \\ 0 \\ -0.8509}u(k), \\
    y(k) &= \bmat{1 & 0 &  0 & 0 \\ 0 & 0 & 1 & 0}\bmat{x_1(k) \\ x_2(k) \\ x_3(k) \\ x_4(k)},
\end{align*}
where $x_1$ and $x_3$ represent the angular positions (rad) of the first link and the second link (relative to the first link),  $x_2$ and $x_4$ are the corresponding angular velocities (rad/s), $u$ represents the torque (Nm) applied on the first link, and $y$ is the plant output.

At training, we pick $\rho = 0.98$. The observation of output is limited to $x_1: [-1, 1], x_3: [-1, 1]$. The hyperparameters are set to $n_\phi = 16, n_\xi = 16$. The following reward is used,
$$
R = \sum_{k=0}^T \left[5.0 - x_{1}(k)^2 - 0.05 x_{2}(k)^2 - x_{3}(k)^2 - 0.05 x_{4}(k)^2 - 0.2 u(k)^2 \right],
$$
where $(x_{1}(k), x_{2}(k), x_{3}(k), x_{4}(k))$ is state and $u(k)$ is control at step $k$.

\subsection{Vehicle Lateral Control}\label{app:veh}
In this setting, we consider the vehicle lateral control problem from \citep{alleyne1997comparison, yin2021stability}. The goal is for the vehicle to track the lane edge while avoiding strong control inputs. The continuous-time linear vehicle lateral dynamics are
\begin{align*}
\scriptsize{\bmat{\dot{e} \\ \ddot{e} \\ \dot{e}_\theta \\ \ddot{e}_\theta} = \bmat{0 & 1 & 0&0 \\ 0 & \frac{C_{\alpha f} + C_{\alpha r}}{m U} & -\frac{C_{\alpha f} + C_{\alpha r}}{m} & \frac{a C_{\alpha f} - b C_{\alpha r}}{m U} \\ 0 & 0 & 0&1 \\ 0 & \frac{a C_{\alpha f} - b C_{\alpha r}}{I_z U} & -\frac{a C_{\alpha f} - b C_{\alpha r}}{I_z} & \frac{a^2 C_{\alpha f} + b^2 C_{\alpha r}}{I_z U}} \bmat{e \\ \dot{e} \\ e_\theta \\ \dot{e}_\theta}
+\bmat{0 \\ -\frac{C_{\alpha f}}{m} \\ 0 \\ -\frac{a C_{\alpha f}}{I_z}} u + \bmat{0\\ \frac{a C_{\alpha f} - b C_{\alpha r} - mU^2}{m} \\ 0 \\ \frac{a^2 C_{\alpha f} + b^2 C_{\alpha r}}{I_z}} c}
\end{align*}
where $e$ is the perpendicular distance to the lane edge (m), and $e_\theta$ is the angle between the tangent to the straight section of the road and the projection of the vehicle's longitudinal axis (rad). Let $x = [e,\dot{e}, e_\theta, \dot{e}_\theta]^\top$ denote the plant state. The control $u$ is the steering angle of the front wheel (rad). The plant output is $y = [e, \ e_\theta]^\top$. The disturbance $c$ is the road curvature ($1/$m). In this task, we consider a constant curvature $c \equiv 0$. The values for the rest of the parameters are given in \citep{yin2021stability}. The controller is synthesized for the discretized dynamics with the sampling time $\delta=0.02$~s.

At training, we pick $\rho = 0.98$. The observation of output is limited to $e: [-10, 10], e_\theta: [-1, 1]$ and is normalized before feeding into the controller. The trajectory terminates when the limit is violated or the length arrives 200. The hyperparameters are set to $n_\phi = 16, n_\xi = 16$. The following reward is used,
$$
R = \sum_{k=0}^T \left[5.0 - 0.01 e(k)^2 - 0.04 \dot{e}(k)^2 - e_{\theta}(k)^2 - 0.04 \dot{e}_{\theta}(k)^2 - \frac{72}{\pi^2} u(k)^2 \right],
$$
where $(e(k), \dot{e}(k), e_{\theta}(k), \dot{e}_{\theta}(k))$ is state and $u(k)$ is control at step $k$.

\subsection{Power System Frequency Regulation}\label{app:powersys}

In this task, we address the distributed control problem for IEEE 39-Bus New England Power System frequency regulation \citep{fazelnia2016convex, jin2020stability} with the decentralized communication topology shown in Fig.~\ref{fig:ieee39} (a). The main goal is to optimally adjust the mechanical power input to each generator such that the phase and frequency at each bus can be restored to their nominal values after a possible perturbation. The linearized and discretized (sampling time $\delta = 0.2$~s) dynamics of the power system are
\begin{align*}
    \bmat{\Omega(k+1) \\ \omega(k+1)} &= \bmat{I_n & \delta I_n \\ -\delta M_p^{-1} L & -\delta M_p^{-1} D + I_n} \bmat{\Omega(k) \\ \omega(k)} + \bmat{0_{n \times n} \\ \delta M_p^{-1}} u(k), \\
    y(k) &= \bmat{I_n, 0_{n \times n}} \bmat{\Omega(k) \\ \omega(k)},
\end{align*}
where $n$ is the number of rotors ($n = 10$ in this example), the states $\Omega, \omega \in \R^{n}$ represent the phases and frequencies of the rotors, $u \in \R^{n}$ represents the generator mechanical power injections, values for inertia coefficient matrix $M_p$, damping coefficient matrix $D$, and Laplacian matrix $L$ are specified in \citep[Section IV]{fazelnia2016convex}. 

Designing an optimal controller for these systems is challenging, because they consist of interconnected subsystems that have limited information sharing. For the case of distributed control, which requires the resulting controller to follow a sparsity pattern, it has been long known that finding the optimal solution amounts to an NP-hard optimization problem in general (even if the underlying system is linear). End-to-end reinforcement learning comes in handy, because it does not require model information by simply interacting with the environment while collecting rewards.

At training, we pick $\rho = 0.98$. The observation of output is limited to $\Omega_i: [-0.5, 0.5]$ and is normalized before feeding into the controller. The trajectory terminates when the limit is violated or the length arrives 200. The hyperparameters are set to $n_\phi = 20, n_\xi = 20$. The following reward is used,
$$
R = \sum_{k=0}^T \left[5.0 - \|\Omega(k)\|^2 - \|\omega(k)\|^2 -  0.2 \|u(k)\|^2 \right],
$$
where $(\Omega(k), \omega(k))$ is state and $u(k)$ is control at step $k$.

\subsection{Wall-clock Performance: Inverted Pendulum}
In Section~\ref{sec:numerical}, we compare our method against policy gradient by giving a plot of reward versus the number of epochs. The idea is to compare the performance of the methods given that the same number of observations are obtained. Another aspect is to compare the performance given the same amount of computation time. We conduct the experiment for the linear inverted pendulum experiment and the result is given in Fig.~\ref{fig:wall}.

\begin{figure}[h]
  \centering
  \includegraphics[width=0.4\textwidth]{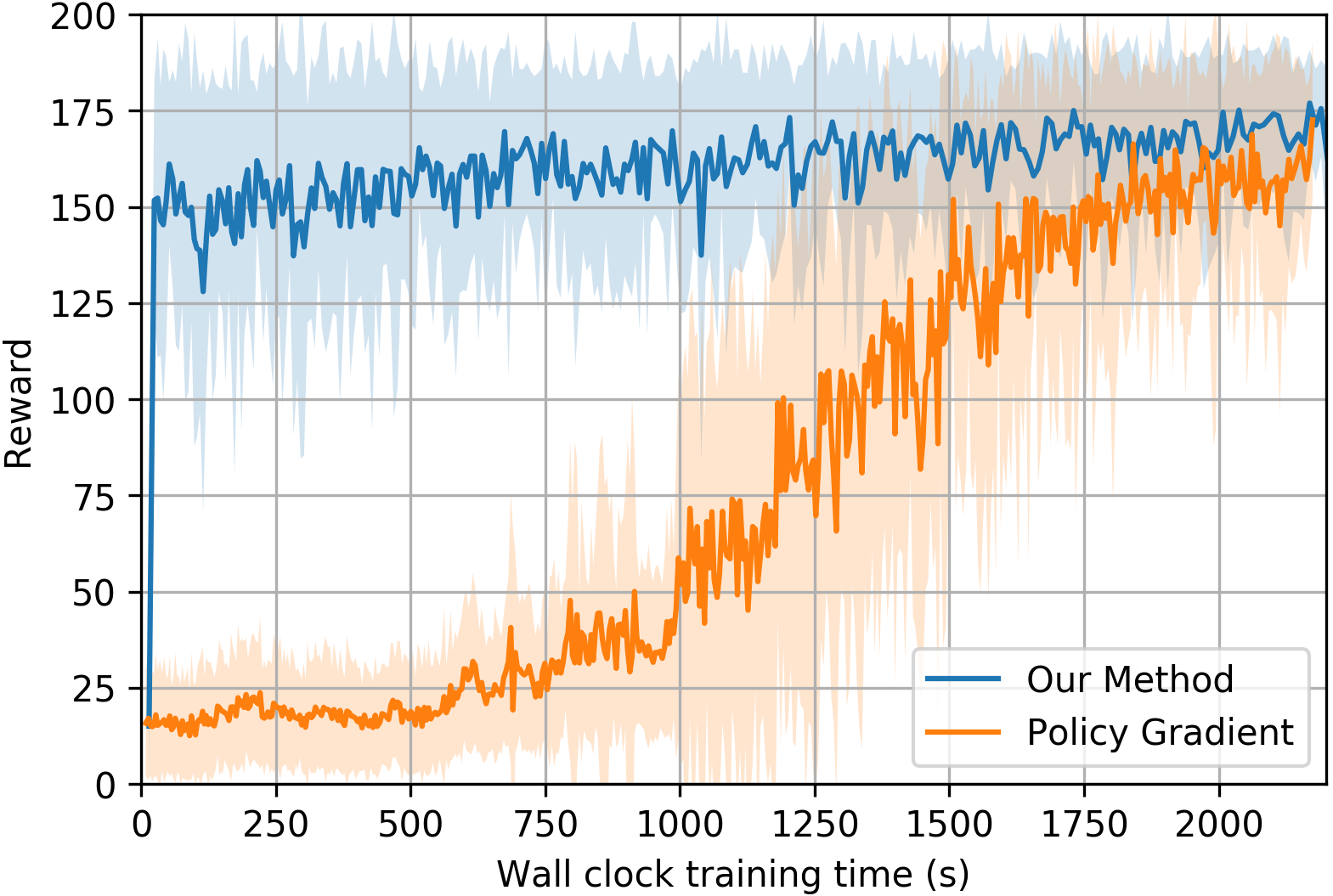}
  \caption{Wall-clock performance for pendulum}
  \label{fig:wall}    
\end{figure}

Although our method is adding a fraction of time for each epoch to conduct the projection step, the overall trend remains the same and the benefits from the ensured stability strictly out-weights the overhead.
\end{document}